\documentclass[12pt]{amsart}

\usepackage[english]{babel}
\usepackage{amsmath,amssymb,amsthm}
\usepackage[colorlinks=true,citecolor=blue,linkcolor=blue,urlcolor=orange]{hyperref}
\usepackage[round]{natbib}
\usepackage{enumitem}
\usepackage{tikz}
\usetikzlibrary{arrows}

\usepackage{geometry}

\usepackage{mathtools}

\usepackage{todonotes}
\usepackage{xfrac}

\usepackage{setspace}

\usepackage{thmtools}

\usepackage{kantlipsum} 

\usepackage{float}

\usepackage{soul}
\usetikzlibrary{shapes,backgrounds}

\newtheorem{theorem}{Theorem}
\newtheorem{proposition}{Proposition}
\newtheorem{lemma}{Lemma}
\newtheorem{corollary}{Corollary}
\newtheorem{remark}{Remark}
\newtheorem{claim}{Claim}

\theoremstyle{definition}
\newtheorem{definition}{Definition}

\newtheorem{innercustomthm}{Example}

\let\Oldenddefinition\enddefinition
\def\enddefinition{\hfill $\triangleleft$\Oldenddefinition}%

\begin{document}
\title{The Luce Model, Regularity, and Choice Overload}
\thanks{$\dagger$ WZB Berlin. E-mail: daniele.caliari@wzb.eu}
\thanks{$\ddagger$ University of Gothenburg. E-mail: henrik@petri.se}

\author[]{Daniele Caliari$^{\dagger}$  \hspace{0.2em} \& \hspace{0.05em} Henrik Petri$^{\ddagger}$ }

\maketitle

\begin{abstract}

We characterize regularity \citep{block} within a novel stochastic model: the General Threshold Luce model [GTLM]. We apply our results to study choice overload, identified by regularity violations that impose a welfare cost on the decision-maker. Generalizing our characterization results, we identify necessary and sufficient conditions for choice overload within GTLMs and, in doing so, disentangle two well-known causes: low discriminatory power \citep{frick} and limited attention \citep{masa2017}.

\end{abstract}


\section{Introduction}

We consider a decision-maker who chooses as follows: she first selects a menu of alternatives and then randomizes over the selected alternatives using logit rules. This generalized version of the standard Luce model \citep{luce59} is called the General Luce Model [GLM].\footnote{GLMs have been extensively studied in the recent literature: \cite{ahumada2018luce}, \cite{echenique2019general}, \cite{cerreia2021canon}, \cite{horan2021stochastic}, \cite{dougan2021odds}, \cite{petri}, \cite{efegerelt}, \cite{rodrigues2024stricter}, \cite{alos2025characterization}.} In this paper, we provide behavioral foundations for regular stochastic choices, an objective that, as we will show, can be achieved within GLMs. Then, we apply our results to choice overload which is often revealed by regularity violations \citep{iyengar2000choice, iyengar2010choice, scheibehenne2010can, dean2024bettertestchoiceoverload}. We provide necessary and sufficient conditions for the identification of choice overload and disentangle two well-known causes: low discriminatory power \citep{frick} and limited attention \citep{masa2017}.

Regularity \citep{block} states that the probability of choosing an alternative (weakly) decreases when another alternative is added to the menu. While it is well-known that the Luce Model satisfies regularity, this is not necessarily true for GLMs. Our particular interest in regularity depends on three factors: (i) its ubiquity in the stochastic choice literature where it is considered "possibly the most well-known property" \citep{cerreiadeliberate}, (ii) the absence of a convincing behavioral interpretation, as we will illustrate in the literature section, and (iii) its role in revealing choice overload. 

The paper opens with a motivating result: a decision-maker has regular stochastic choices for all logit rules if and only if she is a utility maximizer [Theorem \ref{thm: 1}]. A closer look at the proof shows that the potential disagreement between first-stage choices and logit rules imposes demanding rationality conditions. To illustrate, an alternative $x$ may violate regularity because another alternative $y$, discarded in the first stage, is valued infinitely more than $x$ in the second stage (logit rule).\footnote{This observation is not novel within GLMs, see \cite{echenique2019general}, \cite{horan2021stochastic} and it is reminiscent of violations of regularity commonly reported by attention models such as \cite{masa2012}, \cite{mariotticonsideration}, \cite{masa2017}, \cite{masa2020}.} This argument carries a natural follow-up question: under which conditions does regularity hold if logit rules do not contain contradictory information w.r.t. the first-stage choices?

To answer this question, we introduce a new model: the General Threshold-Luce Model [GTLM]. The decision-maker selects the alternatives with a utility higher than a menu-dependent threshold ("good enough") and then randomizes using logit rules aligned with the first-stage choices. GTLMs combine the intuitions from our opening result with two well-known results from the literature. We say $x R y$ if facing a menu with both $x,y$ available, the decision-maker selects $x$ but not $y$. By \cite{fishburn1986axioms}, we know that a utility function that agrees with $R$ - $xRy$ implies $u(x)>u(y)$ - exists if and only if $R$ is acyclic, and by \cite{aleskerov}, we know that $R$ is acyclic if and only if the choices have a menu-dependent threshold representation. In GTLMs, this representation is paired with logit rules in which the utility functions, being aligned with $R$, do not convey the type of contradictory information uncovered in our opening result.

We provide two main characterization theorems for regularity. First, we show that there is an (aligned) logit rule that satisfies regularity if and only if the first-stage choices are path-independent (see \cite{plott1973path}, \cite{moulin}, \cite{yokote2023representation}), [Theorem \ref{thm: existence}]. Second, we show that regularity is satisfied for all (aligned) logit rules if and only if the first-stage choices also satisfy a novel property - denoted $\theta$ - which states that, when a decision-maker faces two menus, one subset of the other, and chooses common alternatives, then she should choose more alternatives from the bigger menu\footnote{Properties that introduce a constraint on the cardinality of choice sets are not new in the literature, e.g. see the model of "order-k rationality" \citep{rozen} and the law of aggregate demand (\cite{hatfield2005matching}, \cite{yokote2023representation}). Furthermore, regular uniform logit rules that we characterize with property $\theta$ were discussed, among many, by \cite{becker} and more recently by \cite{cerreia2018law}.} [Theorem \ref{theorem: R}].

Further analysis of Theorems \ref{thm: existence} and \ref{theorem: R} uncovers the role of the cardinal properties of the utility function in inducing regularity. We introduce novel definitions of concavity/convexity, denoted strong concavity/convexity, i.e. given a linear order $\succ$, a utility function is strongly concave if for any three ordered elements $x \succ y \succ z$, $u(x) - u(y) \leq u(y) - u(z)$, while it is strongly convex if $u(x) - u(y) \geq u(y) - u(z)$. Whenever the three alternatives are consecutive in the linear order then strong concavity/convexity reduces to standard concavity/convexity as defined by \cite{murota}.\footnote{Strongly concave/convex utility functions have been heavily employed in the literature but, to the best of our knowledge, these properties had not been formalized before.} We show that there are logit rules with strongly convex utility functions that always satisfy regularity. In contrast, there are logit rules with strongly concave utility functions that always violate regularity \emph{whenever property $\theta$ is violated}. This result clarifies the role of property $\theta$. In Theorem \ref{theorem: R}, property $\theta$ turns the existential quantifier into a universal quantifier and, by doing so, it ensures that regularity is satisfied in representations with strongly concave utility functions allowing, for example, the decision-maker to optimize the same utility function in both stages.

We apply our results to study choice overload. Our first step is to provide an empirical definition of choice overload within GLMs identifying it with regularity violations that induce a welfare cost to the decision-maker. Combining this novel definition with the logic underlying our characterization theorems, we show [Theorem \ref{thm: overload}] that within GTLMs choice overload occurs if and only if Sen's property $\alpha$ \citep{chernoff, sen71} is violated. Theorem \ref{thm: overload} implies, on the one hand, that the shape of the utility function becomes irrelevant in the analysis of choice overload. On the other hand, it calls for the analysis of two influential models of choice overload that are founded on violations of Sen's property $\alpha$: (i) the monotone threshold model \citep{frick} in which decision-makers discriminatory power worsens in larger menus, and (ii) the limited attention model \citep{masa2017} in which, instead, the attention span deteriorate when more alternatives become available. Leveraging our characterization of choice overload, we show that these two causes are mutually exclusive and provide testable conditions for disentangling them. First, we show that within GTLMs choice overload must be due to low discriminatory power \citep{frick} as the limited attention model exists if and only if Sen's property $\alpha$ is satisfied. Second, we show that outside GTLMs limited attention is the only explanation for choice overload and that this remains true even when Sen's property $\alpha$ is satisfied.

\subsection{Related Literature}

Regularity is most prominently recognized for being a necessary condition for Random Utility models [RUM] \citep{block} and for its fundamental role in the characterization of several influential special cases of RUM such as single-crossing RUM \citep{SCRUM}, Dual RUM \citep{mariottidual}, and the Random Expected Utility model of \cite{gul}. More recently, \cite{sprumont2022regular} also shows that regularity induces every collection of binary probabilities that satisfy the triangle inequality, a known necessary, but not sufficient, condition for binary probabilities to be rationalizable by a RUM. Its role, however, is not limited to RUM. In a lottery domain, \cite{cerreiadeliberate} show that a decision-maker deliberately randomizes using convex preferences if and only if she violates regularity. \cite{masa2020} show that violations of regularity characterize the preferences underlying Random Attention Models. Finally, Additive Perturbed Utility models \citep{fudenberg}, as well as their weaker menu-invariant version \citep{fudenberg14}, have regularity as a necessary condition.

Although regularity is ubiquitous, it lacks a convincing behavioral interpretation. Authors have often argued about regularity through the lens of deterministic properties that are considered analog.\footnote{The seminal papers are \cite{fishburn73} and \cite{fishburn78}; while more recent relevant contributions are \cite{ribeiro2020comparative}, \cite{efegerelt}, and \cite{ok2023measuring}.} More specifically, property $\alpha$ is usually considered the deterministic analog of regularity and, therefore, this latter carries the behavioral interpretation of the former, e.g. "Regularity [Monotonicity] is the stochastic analog of Chernoff Postulate 4, or equivalently, property $\alpha$" \citep{gul}; "regularity is the stochastic generalization of the familiar Chernoff's condition for deterministic choice functions" \citep{dasgupta}; or again, "regularity... is often seen as the stochastic equivalent of independence of irrelevant alternatives (IIA)" \citep{cerreiadeliberate}. In Theorems \ref{thm: existence} and \ref{theorem: R}, we show that the behavioral foundations of regularity go beyond those of property $\alpha$. At the same time, however, in Theorem \ref{thm: overload} we rationalize the common belief that regularity is the stochastic analog of property $\alpha$ because, within our model, property $\alpha$ fully characterizes a decision-maker for whom either regularity holds or its violations are never welfare detrimental.

Our characterization of regularity is routed in the literature on utility representations for boundedly rational decision-makers. Here, we would like to mention the invaluable literature reviews by \cite{aizerman} and \cite{aleskerov}, and with them all the authors who contributed to this beautiful literature. More concretely, we would like to mention a few important papers. Beyond the foundational work of \cite{aleskerov}, \cite{tyson} and \cite{frick} study specific cases of the threshold model that characterize our first-stage choices. \cite{manzini2013two} characterize a deterministic two-stage model where the first stage mirrors ours, and the second stage is the simple maximization of a utility function not necessarily aligned with the choices in the first stage. \cite{ryan2014path} characterizes path-independence in the context of preferences for opportunity while \cite{yokote2023representation} provide a different representation based on a property called "ordinal concavity". 

Finally, recent literature has focused on the characterizations of several instances of GLMs (see the Online Appendix of \cite{horan2021stochastic} for a technical review). This literature differs naturally from our work because of its focus on GLMs and not on regularity per s\'e. To briefly summarize the literature on GLMs, famously \cite{luce59} provides the first characterization of logit rules for strictly positive choice probabilities, i.e. the decision-maker selects all the alternatives in the first stage. \cite{ahumada2018luce} and \cite{echenique2019general} relax the positivity assumption and characterize logit rules imposing no structure on the choices that define the support. \cite{echenique2019general} assume that the choices are rationalizable by a partial order ("two-stage Luce model"), and \cite{echenique2019general} ("threshold Luce model") and \cite{horan2021stochastic} ("stochastic semiorder") also assume that the choices are rationalizable by a semiorder and that the logit rules are aligned. Finally, \cite{ahumada2018luce}, \cite{echenique2019general}, \cite{cerreia2021canon} and \cite{dougan2021odds} require the choice correspondence to be rationalized by a weak order characterizing a model that, as \cite{horan2021stochastic} writes, it is indistinguishable from the Random Utility Model \citep{block} and so satisfies regularity.

\section{Preliminaries}\label{sec:theory}

$X$ is a finite set, $\mathcal{X}$ is the set of all non-empty subsets (menus) of $X$ and $c: \mathcal{X} \to \mathcal{X}$ is a choice correspondence with $c(A) \subseteq A$ and $c(A) \not= \emptyset$ for all $A \subseteq X$. The map $p: X \times 2^{X} \to [0,1]$ with $p(x,A)=0$ if $x \not\in A$ and $\sum\limits_{x \in A} p(x,A) = 1$ for all $A \subseteq X$ is a stochastic choice function. A tie-break rule $\pi$ is a stochastic choice function, such that for all $x \in X$ and $A \subseteq X$: $p_{\pi \vert c}(x, A) = \pi(x, c(A))$.

\subsection{The General Luce Model}

The decision-maker chooses in two stages. First, she selects a menu of alternatives and then uses a tie-break rule to decide which alternative to effectively choose, i.e. the first stage defines the support for the tie-break rule of the second stage. General Luce Models [GLMs] focus on a specific set of tie-break rules: $u$ is a non-degenerate logit rule if there is a function $u: X \to (0,\infty)$ such that for all $x \in A$ and $A \subseteq X$:
$$p_{u \vert c}(x, A) = \begin{cases}
\mbox{  0} & \mbox{if } x \not\in c(A) \\
\frac{u(x)}{\sum_{b \in c(A)} u(b)} & \mbox{if } x \in c(A)\\ 
\end{cases}$$
A special case of logit rules is the uniform rule, denoted $U$, which assumes that all utilities are the same.
$$p_{U \vert c}(x, A) = \begin{cases}
\mbox{  0} & \mbox{if } x \not\in c(A) \\
\frac{1}{\vert c(A) \vert} & \mbox{if } x \in c(A)\\ 
\end{cases}$$
Importantly, notice that there are infinitely many logit rules for each choice correspondence while there is a bijection between choice correspondences and uniform tie-break rules.

\subsection{Regularity and the motivation for a new model}\label{sec: tie-break}

Regularity states that the probability of choosing an alternative $x$ is decreasing in set inclusion. 

\begin{definition}[Regularity]
For all $x \in A \subseteq B$, $p(x,A) \geq p(x,B)$.
\end{definition}

The following result motivates our new model. We show that bounded rationality (i.e. violations of utility maximization) induces violations of regularity for some logit rules. To show this, we introduce the two classical properties - property $\alpha$ and $\beta$ - that characterize utility maximization (see \cite{sen71}).

\begin{definition}[Property $\alpha$]
For all $x \in A \subseteq B$, $x \in c(B)$ implies $x \in c(A)$.
\end{definition}
\begin{definition}[Property $\beta$]
For all $A,B \subseteq X$ and for all $x,y \in c(A)$, if $A \subseteq B$ then $x \in c(B)$ if and only if $y \in c(B)$.
\end{definition}

\begin{theorem} \label{thm: 1}
Within GLMs, first-stage choices $c$ satisfy property $\alpha$ and $\beta$  if and only if $p_{u \vert c}$ is regular for all logit utility functions $u$.
\end{theorem}
\begin{proof}
Let $x \in A \subseteq B$. If $x \not\in c(A)$ then property $\alpha$ implies $x \not\in c(B)$ and so $p_{u \vert c}(x,A) = p_{u \vert c}(x,B)=0$ and regularity holds. If $x \in c(A)$ then either $x \in c(B)$ and by property $\beta$, $c(A) \subseteq c(B)$ and so $p_{u \vert c}(x,A) = \frac{u(x)}{\sum_{y \in c(A)} u(z)} \geq \frac{u(x)}{\sum_{z \in c(B)} u(z)} = p_{u \vert c}(x,B)$. Otherwise, $x \not\in c(B)$ and then $p_{u \vert c}(x,A)>p_{u \vert c}(x,B)=0$. 

Conversely, let $c$ be a choice correspondence and assume $p_{u \vert c}$ is regular for all logit utility functions $u$. Clearly, $c$ satisfies property $\alpha$. We now show that it also satisfies property $\beta$. Assume by contradiction that there are sets $A \subseteq B$ with $c(A) \cap c(B) \not= \emptyset$ but that there is a $y \in c(A)$ with $y \not\in c(B)$. Since $c(A) \cap c(B) \not= \emptyset$ there is a $x \in c(A)\cap c(B)$ with $x \not= y$. Let $\varepsilon>0$ be such that $\varepsilon \vert c(B) \vert < 1 - \varepsilon + \varepsilon \vert c(A) \vert$ and $u$ be a logit utility function with $u(y)=1$ and $u(x)=\varepsilon$ for all $x \in X \setminus y$. Then:
$$p_{u \vert c}(x,A) = \frac{u(x)}{\sum_{z \in c(A)} u(z)} = \frac{\varepsilon}{1 + \varepsilon (\vert c(A) \vert - 1)} < \frac{\varepsilon}{\varepsilon \vert c(B) \vert} = \frac{u(x)}{\sum_{z \in c(B)} u(z)} = p_{u \vert c}(x,B)$$
which contradicts $p_{u \vert c}$ being regular.
\end{proof}

\begin{remark} \label{rem: 1}
The proof highlights an important feature of the logit rules that may violate regularity, i.e. the contradiction is achieved with $x \in c(A) \cap c(B)$ and $y \not\in c(B)$, even if $\frac{u(y)}{u(x)} = \frac{1}{\varepsilon}$ is arbitrarily large. In other words, the decision-maker's preferences that yield the choices in the first stage can be completely misaligned with the utility function that characterizes the logit rule in the second stage. 
\end{remark}

Theorem \ref{thm: 1} raises the question of whether bounded rationality is compatible with GLMs whenever only some logit rules are required to be regular. Building on the discussion in Remark \ref{rem: 1}, we focus on the logit rules that are aligned to the following binary relation: $xRy$ if and only if $x \in c(A)$ and $y \in A \setminus c(A)$. From here, two well-known results complete the motivation for our model. First, following \cite{fishburn1986axioms}, we know that there is a utility function $u$ such that $xRy$ implies $u(x) > u(y)$ ["partial agreement"] if and only if $R$ is acyclic, i.e. the existence of a utility function in the logit rule that respects the revealed preference $R$ requires $R$ to be acyclic. Second, following \cite{aleskerov} we know that $R$ is acyclic if and only if the choices are rationalizable by a menu-dependent threshold model.

\section{The General Threshold-Luce Model}

A General Threshold-Luce Model [GTLM] is characterized by two utility functions $(v, u)$ and a threshold function $\varepsilon: 2^X \to \mathbb{R}^+$. The decision-maker selects the alternatives in the first stage as follows:
$$c(A) = \{ x \in A : \max_{y \in A} v(y) - v(x) \leq \varepsilon(A) \}$$
In the second stage, she applies a non-degenerate logit rule characterized by the utility function $u$ such that $u$ is aligned with $R$, i.e. $xRy$ implies $u(x) > u(y)$. This definition encompasses but is not limited to, the case where $v = u$.

\subsection{The existence of regularity}

Our first main result is a characterization of the existence of regular stochastic choice in GTLMs which relies on the well-known property of path-independence \citep{plott1973path}.

\begin{definition}[Path independence]
For all $A, B \subseteq X$, $c(A \cup B) = c(c(A) \cup c(B))$. 
\end{definition}

\begin{theorem} \label{thm: existence}
Within GTLMs, first-stage choices $c$ satisfy path-independence if and only if there is a utility function $u$ such that $p_{u|c}$ is regular.
\end{theorem}

\begin{proof}
The "if" part is encoded in the following lemma.
\begin{lemma} \label{lemma: pi}
Path-independence is satisfied if and only if there is a tie-break rule $\pi$ such that $p_{\pi \vert c}$ is regular.
\end{lemma}
\begin{proof}
The result follows by combining two results from the literature: (i) $c$ satisfies path-independence if and only if there exists a multi-utility representation such that $x$ is chosen whenever it is the maximal alternative for at least one utility function \citep{aizerman}; (ii) $p$ is regular if and only if there exists a collection of utility functions endowed with a super-additive measure such that the probability of choosing $x$ is equal to the measure on the utilities in which $x$ is a maximal alternative \citep{mcclellon}. We refer to these papers for the formal details.
\end{proof}

We next show the converse. By GTLMs, $c$ has a representation $c(A) = \{ x \in A : \max_{y \in A} v(y) - v(x) \leq \varepsilon(A) \}$ for all $A \subseteq X$. We use a linear extension of the binary relation $R$ to enumerate the $n$ alternatives in $X$ as $\{x_1,...,x_n\}$ and such that $i > j$ if $x_i R x_j$ for all $i,j \in \mathbb{N}$. Define a function $u$ by $u(x_i)=2^i$ for all $i$. First, notice that $u$ satisfies the inequality $u(x_{k+1})  > \sum\limits_{i=1}^{k} u(x_i)$ for all $k \geq 1$. This follows since
$$ u(x_{k+1}) - u(x_{k}) = 2^{k+1}-2^k > 2^k - 2 = \sum\limits_{i=1}^{k-1} 2^i  =  \sum\limits_{i=1}^{k-1} u(x_i).$$ 

We next show that $p_{u|c}$ is regular. Let $x \in A \subseteq B$. We split the proof into three cases. 

\vspace{2mm}
\noindent Case 1: Assume $c(B) \subseteq c(A)$. Path independence implies that $c(B) = c(A \cup B) = c(c(A) \cup c(B)) \supseteq c(c(A)) = c(A)$, so $c(A) =c(B)$. Hence, $p_{u|c}(x,A)=p_{u|c}(x,B)$ and regularity clearly holds. 

\vspace{2mm}
\noindent Case 2: Assume that $c(B) \supseteq c(A)$. In this case regularity trivially holds.

\vspace{2mm}
\noindent Case 3: Assume that $c(B) \setminus c(A) \neq \emptyset$ and $c(A) \setminus c(B) \neq \emptyset$. Then there is an $y \in c(B)\setminus c(A)$ such that $\max_{x \in B} u(x) - u(z) > \varepsilon(B) \geq \max_{x \in B} u(x) - u(y)$ for all $z \in c(A) \setminus c(B)$. This implies that $y$ corresponds to a higher natural number than each of the elements of $c(A)\setminus c(B)$. Hence, by the above reasoning, it must hold that $\sum\limits_{z \in c(B)\setminus c(A)} u(z) \geq u(y) > \sum\limits_{z \in c(A)\setminus c(B)} u(z)$. It thus follows that  $$\frac{1}{p_{u|c}(x,A)}=\frac{\sum_{z \in c(A)} u(z)} {u(x)}= \frac{\sum_{y \in c(A) \setminus c(B)} u(y)}{u(x)} + \frac{ \sum_{y \in c(A) \cap c(B)} u(y)} {u(x)} <$$ 
$$ \frac{\sum_{z \in c(B) \setminus c(A)} u(z)}{u(x)} + \frac{ \sum_{z \in c(A) \cap c(B)} u(z)} {u(x)} = \frac{\sum_{z \in c(B)} u(z)} {u(x)}= \frac{1}{p_{u|c}(x,B)},$$ which implies regularity.\footnote{Notice the importance of the non-degenerate assumption on the logit rule in our proof. If the logit rule is degenerate the restrictions on the choice correspondence weaken considerably. Consider the following property (a well-known weakening of property $\alpha$): $c$ satisfies the fixed-point property if for all $A \subseteq X$ there is a $b \in A$ such that for all $B \subseteq A$ if $b \in B$ then $b \in c(B)$. It is immediate to show that $c$ satisfies the fixed-point property if and only if there is a possibly degenerate logit tie-break rule $u$ such that $p_{u \vert c}$ is regular. Therefore, whenever path independence is violated and the fixed-point property is satisfied the only regular logit rules are the degenerate ones.}
\end{proof}

\subsection{Regularity for everyone}

Our second main result is the characterization of regularity for all decision-makers within GTLMs. Since regularity has to hold for all aligned second-stage choices, it must in particular hold in the knife-edge case when second-stage choices are uniform. To gain insights into the characterization we will therefore begin by analyzing the uniform case. In two simple examples, we show that (i) $p_{U|c}$ violates regularity even if $c$ is highly rational, and (ii) $p_{U|c}$ satisfies regularity even if $c$ is highly irrational. Since a result that covers all decision-makers must also encompass the uniform rule, the desired property must be violated in (i) and satisfied in (ii).

\begin{innercustomthm} \label{ex: ex1}
Let first-stage choices $c$ be maximized by the binary relation $\succ$\footnote{We say that first-stage choices $c$ are maximized by $\succ$ if for all non-empty $A \subseteq X$: $c(A) = \{a \in A : b \succ a \,\, \text{for no} \,\, b \in A\}$.} represented in the digraph below (i.e. $x \succ y$ if and only if there is an arrow from $x$ to $y$).\footnote{This is a simplest semiorder \citep{aleskerov_simplest}, arguably the non-rational binary relation that comes the closest to being rational.} The resulting choice probabilities are: $p_{U \vert c}(x,\lbrace x,y,z,w \rbrace) = 0.5$ and $p_{U \vert c}(x, \lbrace x,y,z \rbrace) = 0.33$, violating regularity.\footnote{\cite{horan2021stochastic} provides a different example using the simple semiorder \citep{aleskerov2002simple} on $\lbrace x,y,z,w \rbrace$: $x \succ y,z$, noticing that $p(w,\lbrace y,z,w \rbrace) < p(w,X)$.}
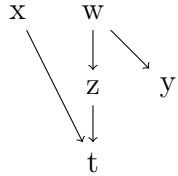
\begin{figure}[H]
\centering
\begin{tikzpicture}[scale=1]
\node (x2) at (12,0) {x};
\node (y2) at (14,-1) {y};
\node (z2) at (13,-1) {z};
\node (w2) at (13,0) {w};
\node (t2) at (13,-2) {t};

\draw [->] (w2)--(y2);
\draw [->] (w2)--(z2);
\draw [->] (z2)--(t2);
\draw [->] (x2)--(t2);
\end{tikzpicture}
\end{figure}

\end{innercustomthm}

In example \ref{ex: ex2}, a choice correspondence that is not rationalizable by any binary relation, and where the order $R$ in GTLMs is cyclic, still induces a $p_{U|c}$ that satisfies regularity.

\begin{innercustomthm} \label{ex: ex2}
Suppose that first-stage choices $c$ are as in the table below. Note that $z \in c(y,z) \cap c(z,w)$ but $z \not\in c(y,z,w)$ implies a violation of binariness \citep{nehring}, or equivalently Sen's $\gamma$ \citep{sen71}. Further, both $z R y$ and $y R z$, so $R$ is cyclic. However, the resulting $p_{U \vert c}$ satisfies regularity.

$$\left[\begin{array}{c|cccccccccccccccc}
A & \lbrace x,y,z,w \rbrace & \lbrace x,y,z \rbrace &\lbrace x,y,w \rbrace & \lbrace x,z,w \rbrace & \lbrace y,z,w \rbrace\\
c(A) & x,y & x,y & x,y & x,z & y,w \\
A & \lbrace x,y \rbrace & \lbrace x,z \rbrace & \lbrace x,w \rbrace & \lbrace y,z \rbrace & \lbrace y,w \rbrace & \lbrace z,w \rbrace\\
c(A) & x,y & x,z & x,w & y,z & y,w & z,w \\
\end{array}\right]$$
\end{innercustomthm} 

The structure of the uniform rule, as also expressed by our two examples, calls for a new property that involves the cardinality of the choice menus, similar to the well-known law of aggregate demand (\cite{hatfield2005matching}, \cite{yokote2023representation}).

\begin{definition}[Property $\theta$] 
$ A \subseteq B$ and $c(A) \cap c(B) \not= \emptyset$ imply $\vert c(A) \vert \leq \vert c(B) \vert$.
\end{definition}

The intuition behind property $\theta$ is simple. As regularity requires choice probabilities to decrease in set inclusion, choice probabilities become more uniform as more alternatives are added to the set. The following result, which anticipates our characterization theorem, clarifies some interesting features of property $\theta$. It shows that property $\theta$ characterizes regularity under uniform logit rules in the most general setting (outside of GTLMs) and, therefore, that within GTLMs property $\theta$ implies the existence of regularity for uniform rules.

\begin{proposition} \label{prop: uniform}
First-stage choices $c$ satisfy property $\alpha$ and $\theta$ if and only if $p_{U \vert c}$ is regular.
\end{proposition}
\begin{proof}
Take a choice correspondence that satisfies property $\alpha $ and $\theta$. Let $A \subseteq B$. By property $\theta$, if $c(A) \cap c(B) \not= \emptyset$, $\vert c(A) \vert \leq \vert c(B) \vert$. There are four possible scenarios: (i) if $x \in c(A) \cap c(B)$ then $p_{U \vert c}(x,A) \geq p_{U \vert c}(x,B)$; (ii) if $x \in c(A)$, $x \not\in c(B)$ then $p_{U \vert c}(x,A) > p_{U \vert c}(x,B )=0$; (iii) if $x \not\in c(A)$, $x \in c(B)$ then property $\alpha$ is violated; (iv) if $c(A) \cap c(B) = \emptyset$, then $c(B) \not\subseteq A$, and regularity is satisfied since for all $x \in c(A)$, if $x \in B$ then $p_{U \vert c}(x,B)=0$. Since $c$ and $p_{U \vert c}$ are related by a bijection, the converse also holds.
\end{proof}

We are now ready to state and prove our second main result in which we show that the intersection of proposition \ref{prop: uniform} and GTLMs allows us to characterize regularity for all decision-makers. To grasp the intuition behind our result, note that within GTLMs, every logit rule that is more discriminating than the uniform rule must, by alignment, discriminate in the direction of the choices in the first stage. This implies that the uniform rule is a worst-case scenario where the utility of worse alternatives is maximized by being equal to the utility of the best alternatives. As a consequence, if regularity holds for uniform rules then it must hold for any other logit rule.

\begin{theorem} \label{theorem: R}
Within GTLMs, first-stage choices $c$ satisfy property $\alpha$ and $\theta$ if and only if $p_{u | c}$ is regular for all logit utility functions $u$.
\end{theorem}
\begin{proof}

We first show that property $\alpha$ and $\theta$ imply that $p_{u|c}$ is regular for all aligned logit rules. Let $u$ be an aligned logit rule with utility $u: X \to (0, \infty)$. Let $x \in A \subseteq B$. Then property $\alpha$ implies that $c(A) \cap c(B) \not = \emptyset$ so that by property $\theta$ we have $|c(A)\setminus c(B) | + |c(A) \cap c(B)| = |c(A)| \leq |c(B)| = |c(B) \setminus c(A)| + |c(A) \cap c(B)|$, i.e. $|c(A) \setminus c(B)| \leq |c(B) \setminus c(A)|$. Thus, there is an injection $f: c(A) \setminus c(B) \to c(B) \setminus c(A)$. Note that $R$ is acyclic, hence $\max_{x \in B} u(x) - u(y) > \varepsilon(B) \geq \max_{x \in B} u(x) - u(f(y))$ for all $y \in c(A) \setminus c(B)$. Hence,
$$\sum\limits_{y \in c(A)} u(y) = \sum\limits_{y \in c(A)\cap c(B)} u(y) + \sum\limits_{y \in c(A) \setminus c(B)} u(y) <$$
$$\sum\limits_{y \in c(A)\cap c(B)} u(y) + \sum\limits_{f(y) \in c(B) \setminus c(A)} u(f(y)) \leq \sum\limits_{y \in c(A)\cap c(B)} u(y) + \sum\limits_{y \in c(B) \setminus c(A)} u(y) = \sum\limits_{y \in c(B)} u(y)$$
The first inequality follows by $u(f(y)) > u(y)$, and the second follows by $|c(A) \setminus c(B)| \leq |c(B) \setminus c(A)|$.
Conversely, let $c$ be a correspondence and assume that $p_{u|c}$ is regular for all logit utility functions $u: X \to (0, \infty)$. It is clear that $c$ satisfies property $\alpha$. It thus suffices to show that $c$ satisfies property $\theta$. Assume, by way of contradiction, that $A \subseteq B$ and $c(A) \cap c(B) \not= \emptyset$ but that $|c(A)| > |c(B)|$. If $c(B) \setminus c(A)=\emptyset$ then regularity is clearly violated. Assume that $c(B) \setminus c(A) \neq \emptyset$. By similar arguments as above, $|c(A) \setminus c(B)| > |c(B) \setminus c(A)|$, hence there is an injection $f: c(B) \setminus c(A) \to c(A) \setminus c(B)$. Let $u$ be a utility function aligned with $R$ such that $u(x) \geq 1$ for all $x \in X$ and $| \max_{x \in X} u(x) - \min_{x \in X} u(x)| < 1/|X|$. Then:
$$\sum\limits_{y \in c(B)} u(y) = \sum\limits_{y \in c(A)\cap c(B)} u(y) + \sum\limits_{y \in c(B) \setminus c(A)} u(y) <$$
$$\sum\limits_{y \in c(A)\cap c(B)} u(y) + \sum\limits_{f(y) \in c(A) \setminus c(B)} \bigg[ u(f(y)) + \frac{1}{|X|}  \bigg] = $$ $$ \sum\limits_{y \in c(A)\cap c(B)} u(y) + \sum\limits_{f(y) \in c(A) \setminus c(B)}  u(f(y)) + \frac{|f[c(A) \setminus c(B)]|}{|X|}  <$$
$$\sum\limits_{y \in c(A)\cap c(B)} u(y) + \sum\limits_{y \in c(A) \setminus c(B)}  u(y) = \sum\limits_{y \in c(A)} u(y)$$

The first inequality follows by $| \max_{x \in X} u(x) - \min_{x \in X} u(x)| < 1/|X|$, while the second inequality follows since the image $f[c(B)\setminus c(A)]$ of $c(B)\setminus c(A)$ satisfies $|f[c(B)\setminus c(A)]| < |c(A)\setminus c(B)|$ (since $f$ is an injection) and since $u(x) \geq 1 > \frac{|f[c(A) \setminus c(B)]|}{X}$ for all $x \in [c(A)\setminus c(B)] \setminus f[c(A) \setminus c(B)] $. Hence, we have proved that a logit rule such that $p_{u |c }$ is irregular exists, a contradiction.
\end{proof}

\subsection{Representations}\label{sec: representations}

In this section, we provide representation results for GTLMs focusing on cardinal properties of the utility functions and, more specifically, we introduce novel notions of concavity and convexity. First, we discuss notions of convexity for generic discrete functions and then we adapt these definitions to our framework. In our subsequent discussion, we will use the shorthand notation $[0,n]$ for sets of the type $\{1,...,n\}$ where $n \in \mathbb{N}$.

\begin{definition}[\citep{murota}]
A utility function $f: [0,n] \to \mathbb{R}^{++}$ is \textbf{concave} if for all $m \in [0,n]$
$$f(m+1) -f(m) \leq f(m) - f(m-1)$$
and it is \textbf{convex} if for all $m \in [0,n]$
$$f(m+1) + f(m) \geq f(m) - f(m-1)$$
\end{definition}

The following is a stricter version of this property.

\begin{definition}
A utility function $f: [0,n] \to \mathbb{R}^{++}$ is \textbf{strongly concave} if for all $m,k,l \in [0,n]$ with $m > k > l$:
$$f(m) - f(k) \leq  f(k) - f(l)$$
and it is \textbf{strongly convex} if for all $m,k,l \in [0,n]$ with $m > k > l$:
$$f(m) - f(k) \geq f(k) - f(l)$$
\end{definition}

The point of the above definitions is that they only permit functions with rapidly decreasing/increasing differences. For instance, the "linear" function $f$ defined by $f(k)=k$ for all $k \in [0,n]$ is concave, but not strongly concave. The following lemma gives some important examples of strictly concave/convex functions.

\begin{lemma} \label{lemma: cc}
For all $a \geq 2$, the utility function $f(n) = k - a^{-n}$ is strongly concave, while $f(n) = k+ a^n$ is strongly convex. In the former case, we assume $k > 1$ such that $f$ is always strictly positive.
\end{lemma}

We are considering utility functions with (finite) domains of alternatives, while the preceding definitions apply to subsets of natural numbers. Next, we will adapt these definitions to our framework. 

\begin{definition} The functions $u,v: X \to (0, \infty)$ are strongly concave if there is an ordering of alternatives $x_1,...,x_n$ such that the function $f: [0,n] \to \mathbb{R}^{++}$ defined by $f(m)=u(x_m)$ or $f(m) = v(x_m)$ for all $m \in [0,n]$ is strongly concave. \end{definition}

Strongly concave/convex utility functions have been frequently employed, e.g. see \cite{aleskerov} for an overview, while in the more recent literature, for instance, by \cite{frick}, \cite{natenzonhe}, and \cite{yokote2023representation}. Strong concavity is also related to the single-peakedness property of preferences (see \cite{apesteguia2017single}, \cite{petri2023binary}, \cite{valkanova2024revealed} among others for some recent applications). In particular, as we show in the appendix, strongly concave utility functions are single-peaked. 

The next result characterizes the relation between regularity and the cardinal properties of the logit utility functions $u$.

\begin{proposition} \label{pro: conc}
Within GTLMs, let path-independence be satisfied. Then:
\begin{enumerate}
    \item \label{part1: conc} there exist strongly convex utility functions $u$ such that $p_{u | c}$ is regular. \label{pro: con1}
    \item \label{part2: conc} if property $\theta$ is violated, there exist strongly concave utility functions $u$ such that $p_{u | c}$ violates regularity. \label{pro: con2} 
\end{enumerate}
\end{proposition}

\begin{proof}

Part (\ref{part1: conc}) follows by inspection of the proof of sufficiency in theorem \ref{thm: existence}. The utility function constructed there is defined as $u(x_i)=2^i$ where $x_i Rx_j$ if and only if $i > j$, and by lemma \ref{lemma: cc} this utility function is strongly convex.  
\

We next show part (\ref{part2: conc}). Define a (strongly) concave utility function by $u(x_i)= k - 2^{-i}$ for all $i \in \{1,...,n\}$ where $k > 1$. Let $m \geq 1$ and let $S$ be a set of $m$ natural numbers and $S'$ another set of natural numbers with $|S'| \leq m-1$ (where $S'$ is not necessarily a subset of $S$ in which case the inequality below is trivial). Then the following holds:
$$\sum\limits_{i \in S} u(x_i)  \geq  \sum\limits_{i=1}^{m} u(x_i) = \sum\limits_{i=1}^{m} \left[ k - 2^{-i} \right]  = mk - \sum\limits_{i=1}^{m} 2^{-i} > mk -1 >  (m-1)k > \sum\limits_{j \in S'} u(x_j).$$
The first inequality holds since we may enumerate $S$ as $\{s_1,...,s_m\}$ s.t. $s_i  \geq i$ for all $i \geq 1$ and hence $k-2^{-s_i} \geq k - 2^{-i}$ for all $i \geq 1$.  The third (in)equality follows because $\sum\limits_{i=1}^{m} 2^{-i} < \sum\limits_{i=1}^{\infty} 2^{-i} = 1$, and the final inequality follows because $|S'| \leq m-1$ and $k-2^{-i} \leq k$ for all $i\geq1$.

Since property $\theta$ is violated there are subsets $A \subseteq B \subseteq X$ such that $c(A) \cap c(B) \neq \emptyset$ and $|c(A)| > |c(B)|$. Let $S = \{i \geq 1: x_i \in c(A)\}$ and $S'= \{i \geq 1: x_i \in c(B)\}$, then $|S| > |S'|$ and by the above inequality it then follows that $$\sum_{z \in c(A)} u(z) =\sum_{i \in S} u(x_i)   > \sum_{i \in S'} u(x_i) = \sum_{z \in c(B)} u(z),$$ which is a violation or regularity. \end{proof}

On the one hand, proposition \ref{pro: conc} shows that there are strongly convex utility functions that never violate regularity and, therefore, completes the reasoning in the sufficiency of theorem \ref{thm: existence}. On the other hand, it shows that there are strongly concave utility functions compatible with regularity only if property $\theta$ is satisfied. This part complements theorem \ref{theorem: R} where, by turning the existential quantifier into a universal one, property $\theta$ limits the role of the threshold as a confounder. This allows regularity and concavity to co-exist harmoniously as concavity can be transferred from the utility function in the first stage to the one in the second stage, i.e. $u=v$.

\section{Application: choice overload}


There is extensive literature that reveals choice overload via violations of regularity, e.g. see \cite{iyengar2010choice}. However, our previous section has shown that regularity may fail even when choice overload plays no role. To reconcile our results with the existing literature, we analyze the welfare effects of regularity violations and identify choice overload only with those violations that have potentially negative consequences for the decision-maker. Finally, we provide testable conditions to disentangle two influential causes of choice overload: (i) discriminatory power represented by the monotone threshold model of \cite{frick}, and limited attention represented by the model of \cite{masa2017}. 

\subsection{Welfare analysis of regularity violations}

We begin by defining a measure of welfare using a standard expected utility approach for discrete choice models \citep{williams1977formation, small1981applied, train2009discrete}. Within GLMs, for all menus $A \subseteq X$: 
$$W(A) = \sum\limits_{a \in A} u(a) p(a, A)$$
where 
$$p(a, A) = \frac{u(a)}{\sum\limits_{b\in c(A)} u(b)}$$
\begin{definition}[Welfare dominance]
For all $A, B \subseteq X$, we say that $A$ \textbf{welfare dominates} $B$ if and only if $W(A) \geq W(B)$ for all logit utility functions $u$.   
\end{definition}
Since within GTLMs, we focus only on utility functions aligned with $R$, the definition naturally holds only for these utility functions. Nonetheless, note that we ensure that the dual-self nature $(u,v)$ of GTLMs does not impact our welfare measure as welfare dominance must hold both if $u = v$ and $u \not= v$. 

Our next step is to define a first-order stochastic dominance relation. Fix a linear extension $\succ$ of $R$ and define the upper contour set of $a$ as $a^{\uparrow} = \{ b \in X: b \succeq a \}$. The cumulative stochastic choice function w.r.t. $\succ$ becomes:
$$\Gamma^{\succ}_{p}(a,A) = \sum\limits_{b \in a^{\uparrow}} p(a,A)$$
Using this definition and letting $p(A)$ indicate the stochastic choice function defined on the menu $A$:
\begin{definition}[FOSD]
For all $A, B \subseteq X$, we say $p(A)$ first-order stochastically dominates $p(B)$ w.r.t. $\succ$ if 
$$\Gamma^{\succ}_{p}(a, A) \geq \Gamma^{\succ}_{p}(a, B) \hspace{0.5em} \text{for all $a \in X$}$$
and we say that $p(A)$ first-order stochastically dominates $p(B)$ if this is true for all linear extensions $\succ$ of $R$.
\end{definition}
The following lemma combines our definition of welfare-dominance relation to first-order stochastic dominance.

\begin{lemma}[\cite{fishburn1974convex}, Theorem 1] \label{lemma_fosd}
The following are equivalent. For all menus $A, B \subseteq X$,
\begin{enumerate}
    \item $p(A)$ first-order stochastically dominates $p(B)$.
    \item $W(A) \geq W(B)$ for all logit utility functions aligned with $R$.
    \item $A$ welfare dominates $B$.
\end{enumerate}
\end{lemma}

Leveraging lemma \ref{lemma_fosd}, we can now distinguish violations of regularity that are "welfare-increasing" and "welfare-decreasing".

\begin{definition}
Given two menus $A \subseteq B$, we say that a violation of regularity is \textbf{welfare-increasing} if and only if $B$ welfare dominates $A$. Otherwise, we say that it is \textbf{welfare-decreasing}.      
\end{definition}
    
In other words, we apply the conservative approach of considering welfare improving only violations of regularity that leave every decision-maker better off. Since these violations can hardly be a result of limited cognitive resources, we identify choice overload only with violations of regularity that are welfare-decreasing. 

\begin{definition}
Within GLMs, we say that a stochastic choice function $p$ reveals choice overload if and only if there is a welfare-decreasing violation of regularity.
\end{definition}

Our main result of this section is a characterization of choice overload within GTLMs. Note that throughout the proof, we take special care in specifying whether a property called the Outcast condition \citep{aizerman1981general, aleskerov} is violated. 
\begin{definition}[Outcast condition]
For all $A, B \subseteq X$, if $c(B) \subseteq A \subseteq B$ then $c(A) = c(B)$.
\end{definition}
We do so because it is well-known that path independence is equivalent to the intersection between property $\alpha$ and the Outcast condition \citep{aizerman, moulin} and our result will allow us to understand how different failures of path independence have different implications for the welfare of the decision-maker.

\begin{theorem} \label{thm: overload}
Within GTLMs, property $\alpha$ is violated if and only if $p$ reveals choice overload.
\end{theorem}

\begin{proof}

\noindent First, we show that property $\alpha$ implies that all regularity violations are welfare-increasing. Suppose that regularity is violated, i.e. there are $A \subseteq B$ and $a \in A$ with $\rho(a, B) > \rho(a, A)$. We split into several cases.

\vspace{5mm}

\noindent Case 1a. Assume that $c(B) \subseteq A \subseteq B$ and $c(A) \neq c(B)$ (so the Outcast condition is violated). Property $\alpha$ and $c(B) \subseteq A$ implies $c(B) \subseteq c(A)$. Hence, $c(B) \not= c(A)$ implies that there is an element $a \in c(A) \setminus c(B)$, i.e. $c(A)$ is a strict superset of $c(B)$. Within GTLMs, property $\alpha$ implies that for all $b \in c(B)$ and $a \in c(A) \setminus c(B)$, $v(b) > v(a)$ (this does not need property $\alpha$). This follows because if $a \in c(A) \setminus c(B)$ then $\theta(B) > v(a) \geq \theta(A)$. If $b \in c(B)$ then $v(b) \geq \theta(B)$, hence $v(b) > v(a)$.  We next show that $p(B)$ first-order stochastically dominates $p(A)$. Let  $b_*$ be the $\succ$ minimal alternative in $c(B)$. For all $a \in X$ with $b_* \succ a$ we then have $\Gamma^{\succ}_{p}(a,B) = \sum\limits_{b \in a^{\uparrow}} p(b,B) =1  \geq \Gamma^{\succ}_{p}(a,A).$ Next, note that for all $a \in c(B)$ we have $\{b \in c(B) : b \succsim a\} =\{b \in c(A) : b \succsim a\}$, which follows since $c(B) \subset c(A)$ and since for all $b \in c(A) \setminus c(B)$ and $a \in c(B) $ we have $v(a) > v(b)$. Thus, if $a \in c(B)$ we have

$$\Gamma^{\succ}_{p}(a,B) =  \sum\limits_{b \in a^{\uparrow}} p(b,B) =  \frac{\sum_{b \in c(B) : b \succsim a} u(b) }{\sum_{b \in c(B)} u(b)} \geq $$ $$ \frac{\sum_{b \in c(A) : b \succsim a} u(b) }{\sum_{b \in c(B)} u(b)}  >  \frac{\sum_{b \in c(A) : b \succsim a} u(b) }{\sum_{b \in c(A)} u(b) }= \Gamma^{\succ}_{p}(a,A),$$ 
where the strict inequality follows since $c(A)$ is a strict superset of $c(B)$ and $u$ is strictly positive. 

\vspace{5mm}

\noindent Case 1b. Assume that $c(B) \subseteq A \subseteq B$ and $c(A) = c(B)$, but then it is immediate that $p(B)$ first-order stochastically dominates $p(A)$ (cumulative probabilities, hence welfare, are equal).

\vspace{5mm}

\noindent Case 2. Assume that $c(B) \not \subseteq A$. Since regularity is violated, we must have $c(A) \cap c(B) \neq \emptyset$ (i.e. if $a \in A \subseteq B$ and $\rho(a,B) > \rho(a,A)$ then $a \in c(B)$ hence property $\alpha$ implies that $a \in c(A)$). First note that for all $b \in c(B)$ and $a \in c(A)\setminus c(B)$ we have $v(b) > v(a)$ (by reasoning above).  Further, since regularity is violated for all $a \in c(A) \cap c(B)$ we have \begin{equation} \label{eq: regv} \frac{u(a)}{\sum_{b \in c(B)}u(b)} > \frac{u(a)}{\sum_{b \in c(A)} u(b)}.\end{equation}  Let $b_*$ be the $\succ$ minimal alternative in $c(B)$. For all $a \in X$ with $b_* \succ a$ we then have $\Gamma^{\succ}_{p}(a,B) = \sum\limits_{b \in a^{\uparrow}} p(b,B) =1  \geq \Gamma^{\succ}_{p}(a,A).$ Let $a^* \in X$ with $a^* \succsim b_*$. Note that $u(a^*) > u(b)$ for all $b \in c(A) \setminus c(B)$. This implies that if $c(A) \cap \{ a \in A: a \succsim a^*\} = (c(A) \cap c(B)) \cap \{ a \in A: a \succsim a^*\}$.    Now, if $(c(A) \cap c(B) ) \cap \{ a \in A: a \succsim a^*\}= \emptyset$ then $\Gamma^{\succ}_{p}(a^*,B)= \sum\limits_{b \in {a^*}^{\uparrow}} p(b,B) \geq 0  = \sum\limits_{b \in {a^*}^{\uparrow}} p(b,A) = \Gamma^{\succ}_{p}(a^*,A).$ If  if $(c(A) \cap c(B) ) \cap \{ a \in A: a \succsim a^*\} \neq \emptyset$ then 

$$\Gamma^{\succ}_{p}(a^*,B)= \sum\limits_{b \in {a^*}^{\uparrow}} p(b,B) = \frac{\sum_{b \in c(B) \cap c(A): b \succsim a^* }u(b) +\sum_{b \in c(B) \setminus c(A): b \succsim a^* }u(b)  } {\sum_{b \in c(B) }u(b)  }  \geq $$ $$\frac{\sum_{b \in c(B) \cap c(A): b \succsim a^* }u(b) } {\sum_{b \in c(B) }u(b)  }  > \frac{\sum_{b \in c(B) \cap c(A): b \succsim a^* }u(b) } {\sum_{b \in c(A) }u(b)  } = \Gamma^{\succ}_{p}(a^*,A),$$ where the strict inequality follows since regularity is violated (the inequalities in equation (\ref{eq: regv})).



\vspace{5mm}

\noindent Conversely, suppose that property $\alpha$ is violated, we show that $p(B)$ does not first-order stochastically dominates $p(A)$. If property $\alpha$ is violated, then there is an element $b^* \in c(B) \setminus c(A)$ with $b^* \in A$. In particular, this implies that $u(a) \geq \theta(A) > u(b^*) \geq \theta(B)$ for all $a \in c(A)$. Let $a_*$ be the $\succ$ minimal alternative in $c(A)$. Then $\Gamma^{\succ}_{p}(a_*,A)= \sum\limits_{b \in {a_*}^{\uparrow}} p(b,A) = 1 > 1-p(b^*,B)  \geq \sum\limits_{b \in {a_*}^{\uparrow}} p(b,B) = \Gamma^{\succ}_{p}(a^*,B).$ 
Thus $p(B)$ does not first-order stochastically dominate $p(A)$. By lemma \ref{lemma_fosd} there is a utility function such that $W(A) > W(B)$.
\end{proof}

Theorem \ref{thm: overload} has two immediate consequences: (i) the cardinal properties of the utility function do not affect the welfare of the decision-maker; (ii) combined with theorem \ref{thm: existence}, it shows that property $\alpha$ is satisfied if and only either regularity holds or its violations are welfare-increasing; an observation that sheds new light on the relationship between these two properties. 

Finally, we conclude with two observations, encoded in corollary \ref{cor: alpha1}. The failures of the two components of path independence (property $\alpha$ and the Outcast condition) yield different welfare consequences for the decision-maker. Specifically, the failure of property $\alpha$ negatively impacts welfare while that of the Outcast condition leaves welfare unaffected. Second, we show how these failures map directly into the discriminatory power of the decision-maker as whenever property $\alpha$ fails the discriminatory power deteriorates in larger menus while if only the Outcast condition fails the opposite holds true.

\begin{corollary} \label{cor: alpha1}
Within GTLMs:
\begin{itemize}
    \item if property $\alpha$ is satisfied and the Outcast condition is violated at $A \subseteq B$ then $\varepsilon(A) \geq \varepsilon(B)$;
    \item if property $\alpha$ is violated at $A \subseteq B$ then $\varepsilon(A) < \varepsilon(B)$.
\end{itemize}
\end{corollary}

\subsection{On the causes of choice overload}

\subsubsection{Choice overload by \cite{frick}}

Corollary \ref{cor: alpha1} sketched the consequences of property $\alpha$ on the discriminatory power of the decision-maker in GTLMs. We now investigate this topic in more detail. In a recent paper, \cite{frick} described a decision-maker whose discriminatory power is decreasing in set inclusion, i.e. the decision-maker selection becomes more "permissible" as more alternatives are added to the menu. Let $x Q y$ if and only if there exists a menu $A$ such that $y \in A$ and $c(A) \not\subseteq c(A \cup \{ x \})$.

\begin{claim}[\cite{frick}, Theorem 1] \label{cl: frick}
The binary relation $S = R \vee Q$ is acyclic if and only if there is a pair $(v,\varepsilon)$ that represents $c$ such that $v$ is strongly convex and $\varepsilon(A) \leq \varepsilon(B)$ if $A \subseteq B$.
\end{claim}

An important, but overlooked, feature of Frick's representation is that the utility function is strongly convex suggesting that regularity may always be satisfied within her model [Proposition \ref{pro: conc}] if $u=v$. However, Frick's decision-maker may violate property $\alpha$, and so by theorems \ref{thm: existence} and \ref{thm: overload} she also faces a welfare-decreasing violation of regularity.

The following results complement theorem \ref{thm: overload} providing a representation for property $\alpha$ within GTLMs and analyzing its effect when combined with Frick's model. The first result, which mirrors claim \ref{cl: frick}, shows that property $\alpha$ in GTLMs induces a representation in which the utility function is strongly concave and the threshold is decreasing in set inclusion. 



\begin{proposition} \label{prop: delta}
Within GTLMs, the following are equivalent:
\begin{itemize}
    \item[(i)] property $\alpha$ is satisfied;
    \item[(ii)] there is a pair $(v,\varepsilon)$ that represent $c$ such that $v$ is strongly concave and $\varepsilon(A) \geq \varepsilon(B)$ if $A \subseteq B$.
    \item[(iii)] either regularity is satisfied or its violations are welfare-increasing.
\end{itemize}

\end{proposition}

\begin{proof}
The equivalence between (i) and (iii) is a corollary of theorem \ref{thm: overload}. For the equivalence between (i) and (ii), see Appendix \ref{app: omitted}.\footnote{Claim \ref{cl: frick} and proposition \ref{prop: delta} convey an interesting message. To construct a decreasing (resp. increasing) threshold w.r.t. set inclusion, the insights from the shape of the utility function go in the opposite direction than those from the threshold. To elucidate, on the one hand, choice overload is represented by more alternatives chosen in bigger sets, an intuition captured by the decreasing threshold. However, the utility function is strongly convex, a property that seems to imply the opposite! On the other hand, a strongly concave utility function should again imply that more alternatives are chosen in bigger menus, but the increasing threshold provides the opposite intuition as we will show discussing \cite{masa2017}.
}
\end{proof}

Our second result had already been uncovered by \cite{frick}. It shows the role of property $\alpha$ within her model rationalizing, in retrospect, the co-existence of representations with decreasing and increasing thresholds. 

\begin{claim}[Lemma 1, \cite{frick}] \label{cl: frick2}
Let $c$ have a menu-dependent threshold representation such that $\varepsilon(A) \leq \varepsilon(B)$ if $A \subseteq B$. Then property $\alpha$ is satisfied if and only if $\varepsilon(A) = \varepsilon(B)$ for all $A, B \subseteq X$ or, equivalently, a semiorder representation.   
\end{claim}

We conclude our analysis with a corollary of theorem \ref{thm: overload}: Frick's representation is compatible with a choice overload interpretation if and only if property $\alpha$ is violated. 

\begin{corollary}
Within GTLMs, let $c$ have a Frick's representation. Property $\alpha$ is violated if and only if $p$ reveals choice overload.
\end{corollary}

\subsubsection{Choice overload by \cite{masa2017}}

In another recent paper, \cite{masa2017} modelled a decision-maker whose attention span is reduced in large menus causing her to miss potentially good alternatives. We call this model, the Limited Attention Model [LAM].

LAMs are characterized by a pair $(\succeq, \Gamma)$ where $\succeq$ is a weak order and $\Gamma : 2^X \to 2^X$ with $\Gamma(A) \not= \emptyset$ for all $A \subseteq X$ is a consideration function, also called competition filter. Competition filters have the following property: for all $x \in A \subset B$, if $x \in \Gamma(B)$ then $x \in \Gamma(A)$. The decision-maker selects the alternatives as follows:
$$c(A) = \{ x \in \Gamma(A) : x \succeq y  \hspace{0.2em} \forall \hspace{0.2em} y \in \Gamma(A) \}$$

LAMs, similar to Frick's model, induce violations of regularity via property $\alpha$. However, differently from Frick's model, these violations are driven by the failure to consider the best alternatives given that the worst ones, i.e. those dominated within the weak order $\succeq$, are never chosen. For example, $x \succ y \succ z$, $\Gamma(\{x,y,z\}) = y$, and $\Gamma(\{x,y\}) = \{x,y\}$ imply $c(x,y,z) = y$ and $c(x,y) = x$. 

We will show that these violations are ruled out within GTLMs as the intersection of LAMs and GTLMs is characterized by property $\alpha$. This also implies that (i) the intuition underlying LAMs within GTLMs is represented by a decreasing threshold [Proposition \ref{prop: delta}], i.e. the decision-maker restricts her focus when facing bigger menus but would not miss the best alternatives, formalizing the idea that choice overload is interpreted differently in \cite{frick} and \cite{masa2017}, (ii) and violations of regularity due to limited attention within GTLMs never reveal choice overload.

\begin{proposition} \label{prop: la}
Within GTLMs, property $\alpha$ is satisfied if and only if there is a pair $(\succeq, \Gamma)$ that represents the choices. Further, the representation is such that $x \sim y$ for all $x,y \in X$.
\end{proposition}
\begin{proof}
Suppose $c$ can be represented by $(\succeq, \Gamma)$. We show that if $R$ is acyclic then property $\alpha$ is satisfied. Suppose, by contradiction, that $x \in c(B)$, $x \in A \subseteq B$ but $x \not\in c(A)$. Since $x \in c(B)$, we have that $x \in \Gamma(B)$, and by the property of the consideration function also $x \in \Gamma(A)$. Since $x \not\in c(A)$ and $x \in \Gamma(A)$ there is an alternative $y \in \Gamma(A)$ such that $y \succ x$. Since we can take $y = \max(\succ, \Gamma(A))$, this also implies that $y R x$. However, $x \in c(B)$ and $y \succ x$ imply that $y \not\in \Gamma(B)$, and so $y \not\in c(B)$. This implies $x R y$ violating the acyclicity of $R$. 

Conversely, property $\alpha$ (even outside GTLMs) implies the existence of a LAM that represents $c$. This is obtained by setting $c(A) = \Gamma(A)$ for all $A \subseteq X$ to obtain the consideration function and $x \sim y$ for all $x,y \in X$ to obtain a weak order. 
\end{proof}

Finally, we apply our welfare measure to LAMs exactly as we applied it to GTLMs by focusing on utility functions that agree with the asymmetric part of the weak order $\succeq$, i.e. $x \succ y$ implies $u(x) > u(y)$. In doing so, we obtain the following result.\footnote{A stricter approach could require the utility function to fully represent the weak order, i.e. $u(x) > u(y)$ if and only if $x \succ y$ and $u(x) = u(y)$ if and only if $x \sim y$. This approach would yield the following result: 
\begin{proposition} \label{prop: la2}
Within GLMs, let $c$ have a LAM representation. Property $\alpha$ is violated if and only if $p$ reveals choice overload. 
\end{proposition}
\begin{proof}
The first part of the statement follows from proposition \ref{prop: la3}. 

Conversely, if property $\alpha$ is satisfied then for all $x,y \in X$ if there is a menu $A$ such that $x,y \in c(A)$ then $x \sim y$ which, in turn, implies $u(x) = u(y)$. Therefore, for all $A \subseteq B$ such that there is an alternative $x \in c(A) \cap c(B)$, $W(A) = W(B)$.
\end{proof}
}




\begin{proposition} \label{prop: la3}
Within GLMs, let $c$ have a LAM representation. If property $\alpha$ is violated then $p$ reveals choice overload. If property $\alpha$ is satisfied then $p$ reveals choice overload if only if either the utility function is misaligned with $R$ or $R$ is cyclic, i.e. $p$ is not a GTLMs.
\end{proposition}
\begin{proof}

If property $\alpha$ is violated then we have $A \subseteq B$, $x \in c(B) \setminus c(A)$ and another alternative $y \in B$ such that $y \in c(A) \setminus c(B)$ because if $y \in c(B)$ then this would imply $x \sim y$ and $x \in c(A)$, a contradiction. This observation follows the Weak Revealed Indifference axiom in \cite{masa2017}. Instead, the violation of property $\alpha$ reveals that $y \succ x$, hence $u(y) > u(x)$. This again implies that $c(B)$ does not first-order stochastically dominates $c(A)$ and so that $p$ reveals choice overload.

For the second part, let property $\alpha$ be satisfied. Then, if the utility function is aligned with $R$ by theorem \ref{thm: overload} we know that $p$ does not reveal choice overload. 

Conversely, note that $x \sim y$ if there is a menu $A$ such that $x, y \in c(A)$ since $x,y \in \Gamma(A)$. Hence, the utility function is unconstrained among alternatives that are chosen together. Then, by theorem \ref{thm: 1}, a violation of regularity occurs only if property $\beta$ is violated. Hence, there are two alternatives $x,y \in A \subseteq B$, with $x \in c(A) \cap c(B)$ and $y \in c(A) \setminus c(B)$. To achieve a welfare-decreasing violation of regularity, we set $u(y) > u(x)$ as in the proof of theorem \ref{thm: 1}. Notice that, either $R$ is acyclic, hence the utility is misaligned or $R$ is cyclic.
\end{proof}
Combining propositions \ref{prop: la} and \ref{prop: la3}, we obtain the following corollary.
\begin{corollary}
Within GLMs, let $c$ have a LAM representation. Then, $p$ reveals choice overload if and only if $p$ is not a GTLMs.
\end{corollary}

We conclude by commenting on our proof of theorem \ref{thm: 1} using limited attention. We have shown that if property $\alpha$ is satisfied but property $\beta$ is not, regularity is violated as follows: $A \subseteq B$, $x \in c(A) \cap c(B)$, $y \in c(A) \setminus c(B)$, and $u(y) > u(x)$. In other words, the decision-maker reveals $x R y$ while $y$ is the better alternative. Proposition \ref{prop: la3} shows that this violation can be explained by limited attention as $x \in \Gamma(A) \cap \Gamma(B)$, $y \not\in \Gamma(B)$, $x \sim y$, and $u(y) > u(x)$.

\section{Concluding remarks on rationality, regularity, and cardinal properties of the utility function}

In standard economic theory, rationality is represented by weak orders which are also the basis for several influential GLMs (\cite{ahumada2018luce}, \cite{cerreia2021canon}, \cite{dougan2021odds}, \cite{alos2025characterization}). In theorem \ref{thm: 1}, we have shown that if a decision-maker maximizes a weak order, then regularity is always satisfied irrespective of the shape of the utility function and the alignment between the choices in the first stage and the logit rules. At first glance, theorem \ref{thm: 1} seems to mirror theorem \ref{theorem: R} where property $\theta$ guarantees regularity for all aligned logit rules. However, insights from the literature \citep{gilboa1995aggregation} suggest that this is far from being true. In particular, in the latter case, the shape of the utility function, i.e. concavity/convexity, has descriptive power by being a crucial part of different representations. Instead, in the former case, the shape of the utility function does not have descriptive power because cardinal properties of the utility function become meaningless for rational decision-makers.\footnote{To see this, consider the following example where we assume $\varepsilon>0$: $x_4 \succ x_3 \sim x_2 \succ x_1$ implies $v(x_2)-v(x_1) > \varepsilon$, $v(x_3)-v(x_2) \leq \varepsilon$, and $v(x_4) - v(x_3) > \varepsilon$ violating both concavity and convexity.}

This argument is not new. As anticipated, \cite{gilboa1995aggregation} note that semiorder representations allow to "discuss properties such as concavity or convexity of the utility function, properties that the mathematical idealization of weak orders renders meaningless... the ranking of differences implied by \emph{larger than the just-noticeable difference} versus \emph{not larger than the just-noticeable difference} is naturally given in the original preferences and it suffices for fixing the utility function almost uniquely without the additional assumption that the decision makers can answer questions like \emph{do you prefer $x$ to $y$ more than you prefer $w$ to $z$? in a meaningful and coherent way}". We contextualize this quote within our work with a counter-intuitive observation: regularity has a behavioral interpretation only if the decision-maker is boundedly rational.

\appendix

\section{Further results}

\subsection{Path-independence outside GTLMs}

The following example shows that path-independence is not a sufficient condition for the existence of logit tie-break rules such that $p_{u | c}$ is regular.

\begin{innercustomthm}
Set $X = \{ x,y,z,u,v \}$. Let $\succ_1$, $\succ_2$, and $\succ_3$ be three linear orders such that: $u \succ_1 x \succ_1 v \succ_1 y \succ_1 z$, and $z \succ_2 y \succ_2 x \succ_2 u \succ_2 v$ and $u \succ_3 x \succ_3 y \succ_3 z \succ_3 v$. Consider the choice correspondence defined by:
$$c(A) = \max(\succ_1, A) \cup \max(\succ_2, A) \cup \max(\succ_3, A)$$
$c$ clearly satisfies path-independence \citep{aizerman}. Assume $p_{u | c}$ is regular for some logit utility function $u:X \to (0,\infty)$. Then 
$$\frac{u(z)}{u(y) + u(z) + u(v)} = p(z, yzv) \geq p(z, xyzv) = \frac{u(z)}{u(x) + u(z)}$$
$$\frac{u(u)}{u(x) + u(u)} = p(u, xuv) \geq p(u, xyuv) = \frac{u(u)}{u(y) + u(u)}$$
imply $u(y) \geq u(x) \geq u(y) + u(v)$. Hence $u(v)=0$, which is a contradiction.
\end{innercustomthm}

\begin{remark}
Necessary and sufficient conditions for the existence of a logit tie-break rule such that $p_{u|c}$ is regular can be found by combining our Theorem \ref{thm: existence} - path-independence - with results in \cite{kraft1959intuitive} and \cite{fishburn1986axioms} on additive representations.
\end{remark}

\subsection{Partial orders}

Another influential refinement of GLMs is the "two-stage Luce model" proposed by \cite{echenique2019general}. The authors focus on choice correspondences rationalizable by a strict partial order but treat first and second-stage choices independently. However, partial orders provide a natural basis for an alignment property, i.e. $x \succ y$ implies $u(x) > u(y)$. Two brief comments. First, note that alignment with $\succ$ is a weaker property than alignment with $R$ only conditional on the existence of an acyclic $R$. For example, consider the strict partial order $x \succ y$ and $z \succ w$, and the utility function $u(x) > u(z) > u(w) > u(y)$ that is aligned with $\succ$. One can notice that this utility function is not aligned with $R$ because $y R w$ and $u(w)>u(y)$. Second, partial orders are the least rational binary relation compatible with regularity. This fact follows by combining the following observations: (i) path-independence is equivalent to property $\alpha$ and the Outcast condition (\cite{moulin}); (ii) property $\alpha$, $\gamma$, and the Outcast condition are necessary and sufficient conditions for $\succ$ to be transitive \citep{aizerman}; (iii) our theorem \ref{thm: existence} which proves the necessity of path-independence for regularity. Therefore, if an (acyclic) binary relation that rationalizes the data exists then regularity is satisfied only if transitivity holds which in turn implies that the following results can be equivalently stated by substituting "acyclic binary relation" with "partial order". 

\begin{proposition}
Let $c$ be rationalizable by an acyclic binary relation. Path independence is satisfied if, and only if, there is a logit utility function $u$ aligned with $\succ$ such that $p_{u|c}$ is regular. 
\end{proposition}

\begin{proof}
Order the elements of $X$ as $\{x_1,...,x_n\}$ such that $x_i \succ x_j$ implies $i > j$. Define a utility function $u$, as in the proof of theorem $\ref{thm: existence}$, by $u(x_i)=2^i$. We show that $p_{u|c}$ is regular. It suffices to show that $p_{u|c}(x,A) \geq p_{u|x}(x,A \cup w)$ for all $x \in A$, $w \in X \setminus A$ and $A \subseteq X$. Let $x \in A \subseteq X$ and $w \in X \setminus A$. We split into three cases as in the proof of \ref{thm: existence}. The first two cases follow similar reasoning. We show case 3. Assume that $c(A) \setminus c(A \cup w) \neq \emptyset$ and $c(A \cup w) \setminus c(A) \neq \emptyset$. First, note that $c(A \cup w) \setminus c(A) = \{w\}$. Let $x \in c(A \cup w) \setminus c(A)$ then $x \in A \cup w \setminus A =\{w\}$. Since if not then $x \in A$. Property $\alpha$ and $x \in c(A \cup w)$ then implies that $x \in c(A)$. A contradiction. Thus $c(A \cup w) \setminus c(A) \subseteq \{w\}$ and since $c(A \cup w) \setminus c(A) \neq \emptyset$ we must have $c(A \cup w) \setminus c(A) = \{w\}$. We next show that there is an element $y \in c(A \cup w) \setminus c(A )$ such that $u(y) > u(z)$ for all $z \in c(A) \setminus c(A \cup w)$. To see this, note that for each $z \in c(A) \setminus c(A \cup w)$ there is an $x \in A \cup w \setminus A =\{w\}$ with $x \succ z$, i.e. $w \succ z$. That is, $w \succ z$ for all $z \in c(A) \setminus c(A \cup w)$. Since $\{w\} = c(A \cup w) \setminus c(A)$ we are done.
\end{proof}

\begin{lemma} \label{lemma: theta} If property $\alpha$ holds then property $\theta$ holds if and only if for all $A \subset X$ and $w \in X \setminus A$ if $c(A \cap w) \cap c(A) \neq \emptyset$ then $|c(A\cup w)| \geq |c(A)|$. \end{lemma}

\begin{proof} Necessity is obvious. We next prove sufficiency. The proof is by induction on the cardinality of $B \setminus A$ for $A \subseteq B$. If $|B\setminus A |=1$ then $|c(B)| \geq |c(A)|$ by assumption. Suppose that property $\theta$ holds for all $A \subseteq B$ with $c(A) \cap c(B) \neq \emptyset$ and $|B\setminus A| \leq m-1$. Let $A \subseteq B$ with $c(A) \cap c(B) \neq \emptyset$ and $|B \setminus A|=m$.  Let $w \in B \setminus A$ then $A \subseteq B \setminus \{w\} \subseteq B$. Let $z \in c(A) \cap c(B)$ and note that property $\alpha$ implies that $z \in c(B \setminus w)$. Since $c(A) \cap c(B \setminus w) \neq \emptyset$ and $|(B \setminus \{w\}) \setminus A| \leq m-1$ it follows by induction hypothesis that $|c(B\setminus \{w\}| \geq |c(A)|$. Further, since $c(B \setminus \{w\}) \cap c(B) \neq \emptyset$ it follows by assumption that $|c(B)| \geq |c(B\setminus \{w\}|$. Taken together these two inequalities hence show that  $|c(B)| \geq |c(B\setminus \{w\}| \geq |c(A)|$. The claim is proved.  \end{proof}

\begin{proposition} \label{prop: succ}
Let $c$ be rationalized by an acyclic binary relation $\succ$. Property $\theta$ is satisfied if and only if for all logit utility functions $u$ aligned with $\succ$, $p_{u \vert c}$ is regular.
\end{proposition}

\begin{proof}

Since $c$ is rationalized by an acyclic binary relation $\succ$ it satisfies property $\alpha$. We show that property $\alpha$ and $\theta$ imply that $p_{u|c}$ is regular for all aligned logit rules. Let $u$ be an aligned logit rule with utility $u: X \to (0, \infty)$. Let $x \in A \subseteq X$ and $w \in X \setminus A$. Then property $\alpha$ implies that $c(A) \cap c(A \cup w) \not = \emptyset$ so that by property $\theta$ we have $|c(A)\setminus c(A \cup w) | + |c(A) \cap c(A \cup w)| = |c(A)| \leq |c(A \cup w)| = |c(A \cup w) \setminus c(A)| + |c(A) \cap c(A \cup w)|$, i.e. $|c(A) \setminus c(A \cup w)| \leq |c(A \cup w) \setminus c(A)|$. By property $\alpha$ we have that $c(A \cup w) \cap A \subseteq c(A)$. Thus $c(A \cup w) \setminus c(A) \subseteq \{w\}$.

\vspace{5mm}
\noindent \textbf{Case 1.} If $c(A) \setminus c(A \cup w) = \emptyset$ then $c(A) \subseteq c(A \cup w)$ and regularity clearly holds. 

\vspace{5mm}
\noindent \textbf{Case 2.} If $c(A) \setminus c(A \cup w)\neq \emptyset$ then $|c(A) \setminus c(A \cup w)| \leq |c(A \cup w) \setminus c(A)|$ and $c(A \cup w) \setminus c(A) \subseteq \{w\}$ implies that  $|c(A) \setminus c(A \cup w)| = |c(A \cup w) \setminus c(A)| =1$. Let $v \in A$ be s.t. $\{v\} = c(A) \setminus c(A \cup w)$. Since $ v \in c(A) \setminus c(A \cup w)$ there is an $x \in A \cup w \setminus A =\{w\}$ with $x \succ v$, i.e. $w \succ v$. Hence,
$$\sum\limits_{y \in c(A)} u(y) = \sum\limits_{y \in c(A)\cap c(A \cup w)} u(y) + u(v) <$$
$$\sum\limits_{y \in c(A)\cap c(A \cup w)} u(y) + u(w) = \sum\limits_{y \in c(A \cup w)} u(y),$$
where first inequality follows since $u(w) > u(v)$ (here we use that $u$ is aligned with $\succ$).
Conversely, let $c$ be a correspondence and assume that $p_{u|c}$ is regular for all logit utility functions $u: X \to (0, \infty)$. It is clear that $c$ satisfies property $\alpha$. It thus suffices to show that $c$ satisfies property $\theta$. Lemma \ref{lemma: theta} also implies that we only need to check property $\theta$ for menus $A \subseteq B$ with $|B \setminus A|=1$. Assume, by way of contradiction, that $A \subseteq X$, $w \in X \setminus A$ and $c(A) \cap c(A \cup w) \not= \emptyset$ but that $|c(A)| > |c(A \cup w)|$. By similar arguments as above, $|c(A) \setminus c(A \cup w)| > |c(A \cup w) \setminus c(A)|$. Now, there are two cases again. Either $c(A \cup w) \setminus c(A) = \emptyset$ and regularity is clearly violated (since $|c(A) \setminus c(A \cup w)|>0$ implies that $c(A \cup w)$ is a strict subset of $c(A)$). If $c(A \cup w) \setminus c(A) \neq \emptyset$ then $c(A \cup w) \setminus c(A) = \{w\}$.  Let $u$ be a utility function aligned with $\succ$ such that  $v(x) \geq 1$ for all $x \in X$ and $| \max_{x \in X} u(x) - \min_{x \in X} u(x)| < 1/|X|$. Fix some $v \in c(A ) \setminus c(A \cup w)$. Then:
$$\sum\limits_{y \in c(A \cup w)} u(y) = \sum\limits_{y \in c(A)\cap c(A \cup w)} u(y) + u(w) <$$
$$\sum\limits_{y \in c(A)\cap c(A \cup w)} u(y) + u(v) + \frac{1}{|X|}  < $$ $$ \sum\limits_{y \in c(A)\cap c(A \cup w)} u(y) + \sum\limits_{y \in c(A) \setminus c(A \cup w)} u(y)= \sum\limits_{y \in c(A )} u(y)$$

The first inequality follows since $u(w) \leq u(v) + \frac{1}{|X|}$ and the second inequality follows since  $|c(A) \setminus c(A \cup w)| > 1$ and $u(y) \geq 1 > \frac{1}{|X|}$ for all $y \in c(A ) \setminus c(A \cup w)$. Hence, we have proved that a logit rule such that $p_{u |c }$ is irregular exists, a contradiction.
\end{proof}

\begin{remark}
Property $\theta$ provides regularity for all aligned logit rules both in GTLMs and when there is a partial order that rationalizes the choices. However, the proofs of theorem \ref{theorem: R} and proposition \ref{prop: succ} are different in their structure. Specifically, theorem \ref{theorem: R} holds over any collection of menus while Proposition \ref{prop: succ} may fail if some menus are not observed. To clarify, suppose we observe only two menus $B=\{x,y,z,w,t\}$ and $A=\{z,w,t\}$, $c(B) = x,y,t$ and $c(A) = z,w,t$. These choices are compatible with the partial order $y \succ w,z$ which allows the utility function $u(y) > u(z),u(w) > u(x)$. Here, if $u(y) + u(x) < u(z) + u(w)$, regularity is violated because $p_{u|c}(t,B) > p_{u|c}(t,A)$. A sufficient richness condition is the observability of all non-empty menus which in this example would provide a violation of property $\theta$ as $|c(x,t,z,w)| > |c(B)|$. This detail connects our paper to the recent work of \cite{alos2025characterization, rodrigues2024stricter} who characterize GLMs over arbitrary collections of menus.
\end{remark}

\subsection{Strong concavity/convexity and single-peakedness/dippedness}

\cite{petri2023binary} calls a preference $P$ quasi-concave w.r.t. a linear order $\succ$ if for all $x,y,z \in X$ s.t. $x \succ y \succ z$ or $z \succ y \succ x$ it holds that $y P x$ or $y P z$. Quasi-convexity is defined by modifying the consequent to $x P y$ or $z P y$. Quasi concavity is equivalent to single-peakedness and quasi convexity is equivalent to single-dippedness (see \cite{apesteguia2017single}, \cite{petri2023binary}).\footnote{As shown by the authors, single-peakedness is also equivalent to the novel definition of convexity introduced by \cite{richter2019convex}.} Call a preference $P$ strongly concave w.r.t. an order $\succ$ if $P$ is represented by a strongly concave utility function $v$ w.r.t. $\succ$. 

\begin{proposition}  If a preference $P$ strongly concave/convex w.r.t. $\succ$ then it is quasi concave/convex w.r.t. $\succ$. \end{proposition} 

\begin{proof} Assume $P$ is strongly concave w.r.t. $\succ$. Let $v$ be a strongly concave utility representation of $P$. If $x \succ y \succ z$  then $v(x)-v(y) \leq v(y) - v(z)$, and hence $v(x) + v(z) \leq 2v(y)$. Assume that $x P y$ and $z P y$ then  $v(x) > v(y)$ and $v(z) > v(y)$, so $v(x) + v(z) > 2v(b)$. A contradiction. If $z \succ y \succ x$ then $v(z)- v(y) \geq v(y) - v(x)$ and we have a contradiction by similar reasoning as before. \end{proof}

\section{Omitted Proofs} \label{app: omitted}

\subsection{Proof of Proposition \ref{prop: delta}}

The following observation is immediate and will be used frequently below.

\begin{lemma} A function $v: X \to (0,\infty)$ is \emph{strongly concave} if and only if there is a strict linear order $\succ$ s.t. for all $x,y,z \in X$: if $x \succ y \succ z$ then $v(x)-v(y) \leq v(y) - v(z)$.  \end{lemma}

\noindent We first show that monotonicity of $\varepsilon$ implies property $\alpha$. Let $x \in B \subseteq A$ and $x \in c(A)$ then $\max_{y \in A} u(y) - u(x) \leq \varepsilon(A)$ and hence $x \in c(B)$.

\vspace{5mm}
\noindent Conversely, let $v: X \to (0, \infty)$ be a function, such that for all $x,y \in X$: $x R y$ implies $v(x) > v(y)$, and such that $v$ is strongly concave w.r.t. $R$.   

\vspace{5mm}

\noindent To simplify the notation, for all menus $A \subseteq X$, we define $x^{A} = \arg\max_{x \in c(A)} v(x)$, $x_{A} = \arg\min_{x \in c(A)} v(x)$. Note that acyclicity of $R$ implies that $v(x^A)= \max_{y \in A} v(y)$ for all $A \subseteq X$. We will use this fact repeatedly below without explicit mention. 

\vspace{5mm}

\noindent The threshold map is constructed recursively, starting with the grand set $X$, and setting $\varepsilon(X) = v(x^X) - v(x_X)$. Suppose that the thresholds have been defined for all sets $B \subseteq X$ with $|B| \geq k$. Let $A \subseteq X$ be a set with cardinality $k-1$ and set $$\varepsilon(A) = \max \{ \max_{B \subseteq X : B \supset A} \varepsilon(B), v(x^X) - v(x_A) \}.$$ By construction of the thresholds, it is now clear that if $A \subseteq B$ then $$\varepsilon(A) \geq  \max_{C \subseteq X : C \supset A} \varepsilon(C) \geq \varepsilon(B),$$ so the threshold map $\varepsilon$ is monotone in set inclusion. 

\vspace{5mm}

\noindent We next show by induction that our representation is satisfied in every menu. 

\vspace{5mm}

\noindent BASE CASE:  As a base case we show that $c(X)=  \{x \in A : \max_{y \in X} v(y) -v(x) \leq \varepsilon(X)\}$. Let $x \in c(X)$ then $v(x) \geq v(x_X)$ and hence $v(x^X) - v(x) \leq v(x^X) - v(x_X) = \varepsilon(X)$. Conversely, if $v(x^X) -v(x) \leq \varepsilon(X)$ then it must be the case that $x \in c(X)$. If not, then $x \notin c(X)$ and $x_X \in c(X)$ implies that $v(x_X) > v(x)$. Hence $v(x^X) - v(x) > v(x^X)-v(x_X) = \varepsilon(X)$. A contradiction.

  \vspace{5mm}
 
\noindent As induction hypothesis assumes that the representation holds for all sets $B \subseteq X$ with $|B| \geq k+1$. Let $A \subseteq X$ be a set with $|A|=k$ and set $y^* = \arg\max_{x \in A \setminus c(A)} v(x)$. We show that  $v(x^A) -v(y^*) >  \varepsilon(A)$ by dividing into two cases: 
 
 \vspace{5mm}
 
\noindent CASE 1: Assume that $ v(x^X) - v(x_A) = \varepsilon(A)$. Then since $x_A \in c(A)$ and $y^* \notin c(A)$ and since $v$ is a linear extension of $R$, it follows that $v(x_A) > v(y^*)$. By strong concavity of $v$ we hence have that $v(x^A) - v(y^*) > v(x^X) - v(x_A) = \varepsilon(A)$. 


\vspace{5mm}

\noindent CASE 2: Assume that  $\varepsilon(A) > v(x^X) - v(x_A)$ then by construction of the thresholds we have that  $ \max_{C \subseteq X : C \supset A} \varepsilon(B) = \varepsilon(A) > v(x^X) - v(x_A)$. Let $B \subseteq X$ with $v(x^X) -v(x_B)=\varepsilon(B) = \max_{C \subseteq X : C \supset A} \varepsilon(B) $. As a preliminary observation we note that $v(y^*) < v(x_B)$. Suppose that $v(y^*) \geq v(x_B)$ then since $y^* \in A \subseteq B$ it must be that $y^* \in c(B)$ and by property $\alpha$ that $y^* \in c(A)$. A contradiction.

\vspace{5mm}

\noindent We next use the preceding observation to show that $v(x^A) -v(y^*) >  \varepsilon(A)$. But, this follows from the following chain of inequalities $v(x^A) - v(y^*) \geq v(x_B) - v(y^*) > v(x^X) - v(x_B) = \varepsilon(B) = \varepsilon(A)$. The first inequality above holds since $v(x^A) \geq v(x_B)$. To see this, note that $\varepsilon(B) = \varepsilon(A) >  v(x^X) - v(x_A) \geq v(x^B) - v(x_A)$. Since $|B| > |A|$ it then follows by the induction hypothesis that $v(x_A) \in c(B)$ and hence $v(x_A) \geq v(x_B)$. Thus, $v(x^A) \geq v(x_A) \geq v(x_B)$. The second inequality holds since either $v(x^X) =v(x^B)$ and then since $v(y^*) < v(x_B)$ it follows that $v(x^B)- v(y^*) > v(x^X) - v(x_B)$. Otherwise, $v(x^X) > v(x^B)$ and then since $v(x^B) > v(y^*)$ it follows by strong concavity that $v(x_B) - v(y^*) > v(x^X) - v(x_B)$. 

\vspace{5mm}

\noindent Since $v(x^A) -v(y^*) >  \varepsilon(A)$ it follows that $c(A) = \{x \in A : \max_{z \in A} v(z) -v(x) \leq \varepsilon(A)\}$. To see this, let $x \in c(A)$ then $v(x) \geq v(x_A)$ and hence $v(x^A) - v(x) \leq v(x^A) - v(x_A) \leq v(x^X) - v(x_A) = \varepsilon(A)$. Conversely, if $x \notin c(A)$ then $v(x) \leq v(y^*)$ hence $ \max_{z\in A} v(z) - v(x) = v(x^A) -v(x) \geq v(x^A) - v(y^*) > \varepsilon(A)$, so $x \notin \{x \in A : \max_{z \in A} v(z) -v(x) \leq \varepsilon(A)\}$.



\bibliographystyle{apa-good}

\bibliography{bibliography}

\end{document}